\documentclass[11pt,a4paper]{article}
\usepackage[utf8]{inputenc}
\usepackage{amsmath}
\usepackage{amsfonts}
\usepackage{amssymb}

\author{Alexandros Angelopoulos, Eleni Bakali  \\\\ {\small Department of Computer Science,} \\{\small School of Electical and Computer Engineering,} \\{\small National Technical University of Athens, Greece.}\\
\texttt{\{angelop,mpakali\}@corelab.ntua.gr}}
\title{Exact uniform sampling over catalan structures}

\usepackage{hyperref}
\usepackage{tikz}

\usepackage{microtype}

\usepackage{tikz}
\usetikzlibrary{shapes,snakes}
\usetikzlibrary{calc,intersections,through,backgrounds,matrix}
\usepackage{bm}
\usepackage{algorithm}
\usepackage[]{algpseudocode}

\newtheorem{theorem}{Theorem}
\newtheorem{definition}{Definition}
\newenvironment{proof}{\noindent\textit{Proof}. }{\hfill $\Box$}

\def \set#1{\mathbb{#1}}
\definecolor{keyr}{rgb}{0.8,0,0}
\definecolor{keyg}{rgb}{0,0.5,0}
\definecolor{keyb}{rgb}{0,0,0.9}
\definecolor{key}{RGB}{127,40,40}

\newtheorem{proposition}[theorem]{Proposition}
\newtheorem{rul}[theorem]{Rule}
\begin{document}
\maketitle
\begin{abstract}
We present a new framework for creating elegant algorithms for exact uniform sampling of important Catalan structures, such as triangulations of convex polygons, Dyck words, monotonic lattice paths and mountain ranges. Along with sampling, we obtain optimal coding, and optimal number of random bits  required for the algorithm. The framework is based on an original two-parameter recursive relation, where Ballot and Catalan numbers appear and which may be regarded as to demonstrate a generalized reduction argument. We then describe (a) a unique $n\times n$ matrix to be used for any of the problems -the common pre-processing step of our framework- and (b) a linear height tree, where leaves correspond one by one to all distinct solutions of each problem; sampling is essentially done by selecting a path from the root to a leaf - the main algorithm. Our main algorithm is linear for a number of the problems mentioned. 
\end{abstract}

\paragraph{keywords}uniform sampling, optimal coding, catalan structures, convex triangulations, ballot numbers

\section{Introduction}
\label{sec:sampintro}
The consideration of different triangulations of a point set first appears in a well-studied form in the middle of the 18\textsuperscript{th} century by Leonhard Euler: he successfully conjectured the closed formula for the \emph{number of triangulations} of the \emph{convex} $n$-gon, that was what we denote now by $C_{n-2}$, the $(n-2)$-th \emph{Catalan number}. These numbers, named after the Belgian mathematician Eug\`{e}ne Charles Catalan, satisfy the following basic relations: 

\begin{align}
&C_{n+1} = \sum_{k=0}^{n} {C_kC_{n-k}},\quad C_0=1\\
&C_n = \frac{1}{n+1}{2n \choose n},\quad C_0=1 \\
\label{asympt} &C_n \leq \frac{4^n}{n^{3/2}\sqrt{\pi}}\quad  \text{ and } \quad \lim_{n\rightarrow \infty}\frac{4^n}{C_nn^{3/2}\sqrt{\pi}} =1
\end{align}

\noindent They occur as the solution of a very large number of counting problems in combinatorics. \cite{stanley2015} gives more than 200 interpretations, more than 60 exercises, in all, a spectacular volume of work centered around the Catalan numbers.


Some modern views regarding this core problem is the study of algorithms for efficient coding of triangulations, and for uniform generation or sampling a random triangulation. 

\textbf{Motivation-Applications.} Triangulations of polygons are fundamental structures with many applications e.g. in  Computational Geometry \cite{ber00, orou98}, VLSI floorplanning \cite{sai99}, and Graph Drawing \cite{nis04}. Many algorithms like point location, ray-shooting, visibility area computing and computing shortest paths inside simple polygons, are based on triangulations \cite{tar88, gui87a,gui87b,lee83}.
In particular sampling (or equivalently uniform generation) of random triangulations, apart from its theoretical interest, has many practical applications. Interestingly sampling triangulations was first considered by physicists for its applications in the study of two dimensional quantum gravity \cite{amb97}. 
Sampling  plays a central role in  testing and verifying the time complexity  of particular implementations of algorithms like the previously mentioned, that depend on triangulations. A method to achieve time verification is presented in \cite{eps94}. 
Finally efficient coding of triangulations, is essential for storing them after their random generation, since these data are usually huge and inefficient coding may cause storing to dominate the running time of the algorithm that uses these data.

\textbf{Previous work.} There has been significant work on uniform sampling and coding triangulations, especially since the 1990s \cite{eps94ran,ding05ran,soh97enc, dev99ran, pou06sam,sha06tri,sha11deg,sha11ctri}.

We will mostly compare our results and framework to \cite{dev99ran}: Devroye et al.\ describe a linear-time algorithm in the RAM model of computation, for the uniform generation of convex triangulations and monotone lattice paths. This is actually equivalent to a $\mathcal{O}(n\log{n})$ algorithm, when taking into account that this algorithm needs the generation of $O(n)$ random numbers of $O(\log n)$ bits each.
The ballot theorem is invoked in their work (see also \cite{feller1968}, Chapter III), as it describes lattice path problems. 

Another line of research regarding triangulations is the study of Markov chains that move along similar triangulations, most commonly considering as similar those  that differ by a single edge flip \cite{hur99gtri, mol97mix, tet97mix, par11gentri}. These Markov chains converge to the uniform distribution over the set of triangulations, thus yielding \textit{almost} uniform sampling algorithms (this is the well known Markov chain Monte Carlo method for sampling and counting). The best of the above mentioned algorithms runs in $\mathcal{O}(n^5)$, but has a deviation $\epsilon>0$ from the uniform distribution (which is a common drawback/characteristic of the MCMC method).

  While for convex graphs it is the Catalan numbers that give the exact number of triangulations, no such formula exists for a general graph. Exhaustively enumerating the triangulations is not an easy task: already $C_n=\Theta(n^{-3/2}4^n)$, while the best known bounds for the generic case are currently set a lower $\Omega(2.43^n)$ (\cite{sha11deg}) and an upper $\mathcal{O}(30^n)$ (\cite{sha11ctri}).   There is also no general formula for counting triangulations of convex polygons with forbidden edges, i.e. when the input is a general graph G embedded to the plane,  and valid triangulations are those that use only edges of G. Naturally, there has been work on counting the triangulations given a specific point set asymptotically faster than by enumerating all triangulations (\cite{alv13agg}), and approximately counting with favorable compromises (\cite{alv15count}). 
  
\subsection{Our contribution}
In this article, we present a unified framework to obtain algorithms for \textit{exact} uniform \textit{sampling} and \textit{optimal coding} of Catalan structures. We will use the triangulations of convex polygons as our main problem. We will then use the monotone lattice path problem to show that our framework can be tweaked to easily get the very algorithms presented in \cite{dev99ran}. In a nutshell, we present:   

\begin{itemize}
\item new framework and elegant algorithm for exact uniform sampling over convex triangulations and other Catalan structures;
\item new recursive relation for ballot and Catalan numbers with an interesting combinatorial interpretation;
\item separate pre-processing step of $\mathcal{O}(n^2)$ time, common to all problems that can be described by the recursive relation, 
\item optimal coding of the solutions/samples of the problems: for an input related to size $n$, we need exactly $\log(C_n)$ bits to encode each solution;
\item optimal number of random bits for the main sampling algorithm: it needs exactly $\log{C_n}$ random bits to run, i.e.\ to generate uniformly at random a solution;
\item efficient algorithm for sampling, as well as counting triangulations for a more general family of graphs embedded to the plane in convex position: $K_{n,-m}$ is obtained by an embedding of the complete graph $K_{n}$ after removing $m$ consecutive span-2 edges. 
\end{itemize}

Although our sampling algorithm has total running time $O(n^2)$, which is worse than the $O(n \log n)$ of \cite{dev99ran}, our preprocessing step helps to reduce to the optitmal the number of random bits needed, as well as the length of the codewords.

\paragraph*{An outline -  main ideas} 
Given $n$, we consider an initial instance which -in some sense- describes the space of all potential solutions; for instance, in the case of convex triangulations, we consider the complete convex graph $K_n$, as all $n(n-3)/2$ diagonals are in our disposal to be selected for some triangulation, which requires only $n-3$ such edges. Our goal is to describe\footnote{The construction is done only mentally, for the purposes of the analysis; it is not actually performed by the algorithm.} a binary tree of polynomial height, with its leaves corresponding one by one to all triangulations of the convex $n$-gon and then aim at sampling a leaf uniformly at random. 

In order to achieve this, we need a reduction argument, w.r.t.\ which the set of potential solutions on some node's graph can actually be partitioned into the solutions obtained either by its left or the right child and which are disjoint. Then, in order to sample a solution uniformly at random, it would suffice to be able to calculate the size of each of the two subtrees of any node \cite{knuth74, sin89apx}. Thus, starting from the root, we should be capable of uniformly selecting a leaf by recursively selecting one of a node's children with probability proportional to the size of the corresponding subtree.

Usually, the estimation of the size of subtrees for an arbitrary tree is computationally hard \cite{bak17}, and this is the reason that the above method usually fails. The innovation of our approach consists in that we manage to construct this tree in a systematic way that not only reflects our reduction argument, but also plenty of isomorphisms are revealed so that (a) the number of all classes of non-isomorphic instances is polynomial in $n$ (in particular quadratic), and (b) it is easy to determine the class in which each instance belongs. 
Eventually, it suffices to solve the general problem only for a polynomial number of non-isomorphic sub-instances, in a preprocessing step, bottom-up, in time linear to the number of non-isomorphic instances, thus in total quadratic in $n$. 

In particular, it occurs that all internal nodes of our tree correspond to (sub)instances of a more general problem, needing a second parameter to be described and for which no formula nor algorithm was known until now: compute the number of triangulations of an embedded to the plane \emph{almost complete} convex graph $K_{n,-m}$, a graph missing only $m$ \emph{consecutive span-2 edges} (or ears, as denoted in \cite{hur96ear,dev99ran}) from what would be a the complete $K_n$. As a first simple note, we get that $m\leq n$, thus  we have the $\mathcal{O}(n^2)$ classes of non-isomorphic instances.

In the Sections to follow, we will state and prove properties of the central to our work recursive relation, the sampling algorithm scheme it implies and its application to convex triangulations and mountain ranges problems (the latest to be trivially linked also to Dyck words, monotone lattice paths, etc.). 

\section{Our framework}

Let $a_{n,m}$ denote the number of solutions corresponding to any problem parameterized by $n$, $m$, and consider the following recursive relation:

\begin{align}
\label{rel:bc1} a_{n,0}&=a_{n-1,0}+a_{n,1} \\ 
\label{rel:bc2} a_{n,m}&=a_{n-1,m-1}+a_{n,m+1}\\ 
\label{rel:bc3} a_{0,0}&=1 \text{ and } a_{n,n}=0, \, \forall i\geq 1\
\end{align}

\begin{theorem} Consider the ballot numbers $$N_{i,j}=\frac{i+1-j}{i+1+j}{i+1+j \choose j},\, i,j\geq 0$$
Along with defining $N_{i,-1}=0$, the ballot numbers satisfy \ref{rel:bc1}, \ref{rel:bc2} and \ref{rel:bc3} for $i=n$ and $j=n-1-m$.
\label{thm:bal}
\end{theorem}

In other words, $a_{n,m}=N_{n,n-1-m}$. Note also that:
\begin{equation}
a_{n,0}=N_{n,n-1}=\frac{2}{2n}{2n \choose n-1}=\frac{1}{n+1}{2n \choose n}=C_n
\label{eq:bc}
\end{equation}

We will from now on refer to the relations \ref{rel:bc1}, \ref{rel:bc2} and \ref{rel:bc3} as the \emph{BC (recursive) relation or recursion}, a naming derived from its close relationship with ballot and Catalan numbers.

\begin{theorem}
The generating function for the numbers $a_{n,m}$ is $$G(x,y)=\frac{(1-xy)(1-\sqrt{1 - 4x}-2xy)}{2 x (1 - y + x y^2)}$$
\label{thm:gen}
\end{theorem}

\begin{proof}[Proofs of Theorems \ref{thm:bal} and \ref{thm:gen}] The proofs can be found in Appendix \ref{sec:appa}.
\end{proof}

\paragraph*{The BC table}

Immediately, we may define a $n \times n$ table/matrix, with $a_{n,m}$ as entries. It can be calculated in $\mathcal{O}(n^2)$ time, bottom-up (Figure \ref{fig:table}), and will be considered as our framework's pre-processing step. The main algorithm will simply recall some of its values when needed, in constant time. 

\begin{figure}[h]
\centering
\begin{tikzpicture}[ampersand replacement=\&]
\matrix (m) [scale=1.25, matrix of nodes ,row sep=-\pgflinewidth, text width=4.5ex, text height=1.75ex, align=center,anchor=center]
{
\hline
\begin{scope} \tikz\node[overlay] at (-1.125ex,-0.5ex){\scriptsize n};\tikz\node[overlay] at (0.5ex,1ex){\scriptsize m}; \end{scope} \& 0 \& 1 \& 2 \& 3 \& 4  \& 5 \& 6 \& 7 \& 8 \& $\cdots$ \\ \hline
0 \& |[fill=keyg!50]|\textbf{1} \&  \&  \& \& \&  \&  \&  \&  \& \\
1 \& \textbf{1} \& |[fill=keyr!40]| 0 \& \& \& \&  \& \& \&  \&   \\
2 \& \textbf{2} \& 1 \& |[fill=keyr!40]|  0 \& \&\&  \&  \&  \&  \&  \\
3 \& \textbf{5} \& 3 \& 1 \& |[fill=keyr!40]|  0 \&  \& \&  \&  \& \&  \\
4 \& \textbf{14} \& 9 \& 4 \& 1\& |[fill=keyr!40]|  0 \& \&\&  \& \&  \\
5 \& \textbf{42} \& 28 \& 14 \& 5\& 1 \& |[fill=keyr!40]|  0 \& \&  \&\& \\
6 \& \textbf{132} \& 90\&  48 \& 20 \& 6 \& 1 \& |[fill=keyr!40]|  0 \&  \& \& \\
7 \& \textbf{429} \& 297 \& 165 \&  75 \& 27 \& 7 \& 1\& |[fill=keyr!40]|  0   \& \& \\
8 \& \textbf{1430}  \& 1001 \& 572 \& 275 \&  110 \& 35 \& 8 \& 1\& |[fill=keyr!40]|  0   \& \\
$\vdots$  \& $\uparrow$ \& \& \& \& \&  \&  \&  \&  \&  $\ddots$  \\
\& $\bm{C_n}$ \& \& \& \& \&  \&  \&  \&  \&   \\
};


\draw (m-1-1.north east) -- (m-11-1.south east);
\draw (m-1-1.north west) -- (m-11-1.south west);
\draw (m-1-1.north west) -- (m-2-2.north west);


\draw [thick, opacity=0.35, ->] (m-3-2.center) -- (m-4-3.center);
\draw [thick, opacity=0.35, ->] (m-4-3.center) -- (m-4-2.center);

\draw [thick, opacity=0.35, ->] (m-4-2.center) -- (m-5-4.center);
\draw [thick, opacity=0.35, ->] (m-5-4.center) -- (m-5-2.center);

\draw [thick, opacity=0.35, ->] (m-5-2.center) -- (m-6-5.center);
\draw [thick, opacity=0.35, dotted, ->] (m-6-5.center) -- (m-6-3.east);
\end{tikzpicture}
\caption{The BC table. It has $(n-3)(n-1)+1$ entries. The arrows show the sequence in which the cells fill up.}
\label{fig:table}
\end{figure}

\paragraph*{The main algorithm scheme}

Consider now the binary tree implied by the BC relation (Figure \ref{fig:sampalgo}), rooted at a node related to $a_{n,0}$. This means we are attempting to randomly select one of $C_n$ samples. The BC table allows us to achieve:

\begin{itemize}
\item optimal coding, requiring exactly $\log(C_n)=\mathcal{O}(n)$ bits to encode each solution; 
\item optimal number of random bits for the main sampling algorithm: generating one random number in the range $[0..C_n-1]$, which requires $\log{C_n}$ random bits, suffices.
\end{itemize}

A common approach is to calculate the sizes of the two sub-trees hanging from some node \emph{upon arriving on the very node} and branch left or right with the appropriate probability. This requires a total of (number of branchings)$\times$(random bits required per branching) random bits. Indeed, we may well apply the algorithm of \cite{dev99ran}, given that $\frac{a_{n-1,m-1}}{a_{n,m}}\approx 2(n+m)$, the number of random bits needed is $\mathcal{O}(n \log{n})$, as the height of the tree is linear to $n$, thus the branchings required to reach a leaf is linear in $n$, too. 

It is evident that establishing the BC table beforehand, and tossing only $\log{C_n}\approx 4n$ coins (see (\ref{asympt})) we generate a code uniformly at random, which uniquely defines the leaf of the tree the algorithm will end upon: we branch without any more coin tosses, as we know all critical quantities in advance. 

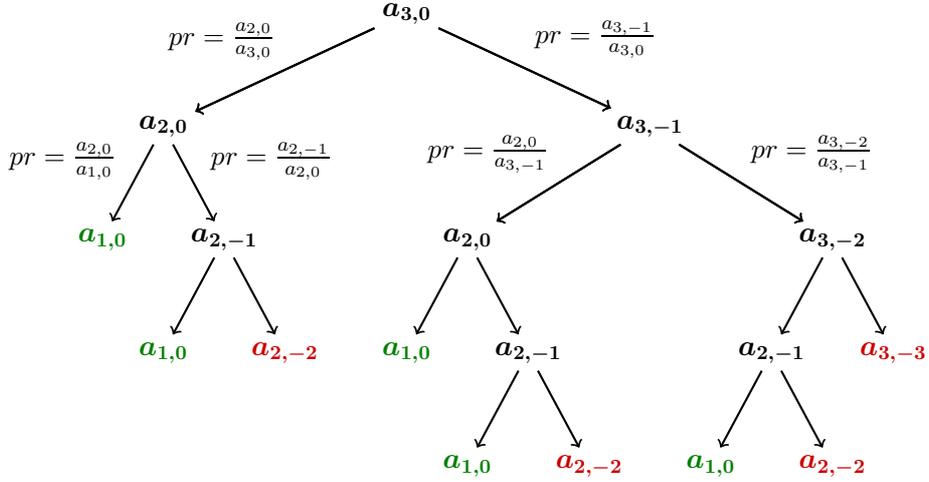
\begin{figure}[h!]
\centering
\begin{tikzpicture}
 [every node/.style={rectangle, draw=none, fill=none, scale=0.75}]

\node (k5) at (0.2,0) {\Large $\bm{a_{3,0}}$};
\node (k54) at (-3,-1.5){\Large $\bm{a_{2,0}}$};
\node (k551) at (3.4,-1.5){\Large $\bm{a_{3,-1}}$};
\node [keyg] (k543) at (-3.8,-3) {\Large $\bm{a_{1,0}}$};
\node (k5441) at (-2.2,-3) {\Large $\bm{a_{2,-1}}$};
\node [keyg] (k54413) at (-3,-4.5) {\Large $\bm{a_{1,0}}$};
\node [keyr] (k544142) at (-1.4,-4.5){\Large $\bm{a_{2,-2}}$};
\node (k5514) at (1,-3) {\Large $\bm{a_{2,0}}$};
\node [keyg] (k55143) at (0.2,-4.5) {\Large $\bm{a_{1,0}}$};
\node (k551441) at (1.8,-4.5) {\Large $\bm{a_{2,-1}}$};
\node[keyg]  (k5514413) at (1,-6) {\Large $\bm{a_{1,0}}$};
\node [keyr] (k55144142) at (2.6,-6) {\Large $\bm{a_{2,-2}}$};
\node (k55152) at (5.8,-3) {\Large $\bm{a_{3,-2}}$};
\node [keyr] (k5515253) at (6.6,-4.5) {\Large $\bm{a_{3,-3}}$};
\node (k5515241) at (5,-4.5) {\Large $\bm{a_{2,-1}}$};
\node[keyg]  (k55152413) at (4.2,-6) {\Large $\bm{a_{1,0}}$};
\node [keyr] (k551524142) at (5.8,-6) {\Large $\bm{a_{2,-2}}$};

\foreach \x/\y in {k5441/k54413, k5441/k544142,k5514/k55143, k5514/k551441, k551441/k5514413, k551441/k55144142, k55152/k5515241,k55152/k5515253, k5515241/k55152413, k5515241/k551524142}
           \draw [thick, ->] (\x)--(\y);     
   \foreach \x/\y in {k5/k54,k5/k551,k551/k5514,k551/k55152}
                \draw [thick, ->] (\x)--(\y);

\draw [thick, ->] (k5)--(k54) node [rectangle, pos=0.5, anchor=south east, scale=1.25] {$pr=\frac{a_{2,0}}{a_{3,0}}$};
\draw [thick, ->] (k5)--(k551) node [rectangle, pos=0.5, anchor=south west, scale=1.25] {$pr=\frac{a_{3,-1}}{a_{3,0}}$};

\draw [thick, ->] (k551)--(k55152) node [rectangle, pos=0.5, anchor=south west, scale=1.25] {$pr=\frac{a_{3,-2}}{a_{3,-1}}$};

\draw [thick, ->] (k551)--(k5514) node [rectangle, pos=0.5, anchor=south east, scale=1.25] {$pr=\frac{a_{2,0}}{a_{3,-1}}$};

\draw [thick, ->] (k54)--(k5441) node [rectangle, pos=0.5, anchor=south west, scale=1.25, yshift=-1mm, xshift=1mm] {$pr=\frac{a_{2,-1}}{a_{2,0}}$};

\draw [thick, ->] (k54)--(k543) node [rectangle, pos=0.5, anchor=south east, scale=1.25, yshift=-1mm, xshift=-1mm] {$pr=\frac{a_{2,0}}{a_{1,0}}$};
                
\end{tikzpicture}
\caption{The universal sampling algorithm scheme: branching with probability analogous to the size of the subtree. The height of the tree is $\mathcal{O}(n)$. We note that typically all $a_{n,n}$ nodes are non-existent, as the algorithm branches towards them with probability 0 - that's why they are marked red. Also, all $a_{1,0}$ labeled nodes have a left child, $a_{0,0}$, which is the actual leaf according to the BC relation.}
\label{fig:sampalgo}
\end{figure}

Moreover, the scheme allows for down to \emph{linear time} main algorithms for sampling an instances, as all leaves' distance from the root is linear in $n$ - it actually fluctuates from $n$ to $2n$. For instance, we will show that the mountain ranges problem (and, therefore, Dyck words and monotone lattice paths) has an $\mathcal{O}(n)$ main sampling algorithm, which we consider a good trade-off when we need many samples for the same problem. Remember, that only a $\mathcal{O}(n \log{n})$ algorithm is proposed in \cite{dev99ran}.

\section{Convex triangulations}

As we have already mentioned, for any problem we consider the object which describes the space of all solutions. In the case of convex triangulations, this is the \emph{complete geometric graph}. A geometric graph is defined \cite{bosemain} as a pair of a point set $V$ on $\set{R}^2$ and a subset of straight line segments with endpoints in $V$, i.e.\ $E \subseteq {V \choose 2}$. Now let our point set be convex and $E = {V \choose 2}$. We may refer to this special case as the \emph{``convex $K_n$''}, an alternate to the complete convex geometric graph, taking into account that all convex drawings of $K_n$ are isomorphic. Some important definitions are the following.

\paragraph*{Proper labeling, consecutive vertices and span-2 edges} 

Given a convex geometrical graph, we will say that its vertex labeling $v_0,...,v_{n-1}$ is \emph{proper} if its $n$ vertices appear in order, clockwise or counter-clockwise around its convex hull. We will always assume such a labeling for any convex graph we refer to, unless noted otherwise. The above being stated, it is evident that for all $i \in [0..n-1]$, vertices $v_i$ and $v_{i+1}$ are consecutive, naturally defining that $v_n \equiv v_0$. We shall also admit that all edges $v_iv_{i+2}$, $i \in [0..n-1]$ are properly defined, as $v_x \equiv v_{x\text{ mod } n}$, for any $x,n$. 

Let $V$ be a convex point set and assume proper vertex labeling. For every edge (diagonal) $e=v_iv_j \in E$, we define its \emph{span}, denoted by $|e|$, to be the minimum distance of its adjacent vertices around the convex hull. The edges that miss only one vertex, that is edges of the form $v_{i-1}v_{i+1}$, will hold a key-role in our work, thus we will denote them as \emph{span-2 edges}. All such edges can be named after the missed vertex with which they form a triangle: $e_i \equiv v_{i-1}v_{i+1}$.

As an analog to the notion of consecutive vertices, we define two \emph{consecutive span-2 edges} to be a pair of edges of the form $v_{i-1}v_{i+1}$, $v_iv_{i+2}$. 

Observe that: $K_3$ has no span-2 edges, $K_4$ has two and, for $n\geq 5$, $K_n$ has exactly $n$ of such edges. For $n\geq 5$ and any vertex $v_i$, we denote by $e_i$ the span-2 edge $v_{i-1}v_{i+1}$, in other words, the edge that may form a triangle together with $v_i$.

\paragraph*{The $K_{n,-m}$ graph.} We will need to work on specific graphs, that allow for the desired structure of triangulations to be witnessed. We define $K_{n,-m}$ to be the nearly complete convex $K_n$ which \emph{misses only $m$ consecutive span-2 edges}. Therefore:
\begin{itemize}
\item $K_{n,0} \equiv K_n$ and when properly defined, $K_{n,-m}$ has ${n \choose 2}-m$ edges.
\item For $n\geq 5$, it is $-n \leq -m \leq 0$, as a graph may miss up to all $n$ span-2 edges. 
\item The graphs $K_{4,-1}$, $K_{4,-2}$ are properly defined. 
\item For fixed $n,m$, all $K_{n,-m}$ are isomorphic to each other (under rotation).

\end{itemize}

\begin{proposition}
Every triangulation of a convex geometric graph on 5 or more vertices has at least two span-2 edges. 
\label{pro:span22}
\end{proposition}

\begin{proof}
A triangulation $T_C$ on $n$ points requires that all faces but the outer are triangles. The number of faces $f$ is equal to $n-1$ (e.g.\ use Euler's characteristic on plane graphs), therefore there is a total of $n-2$ inner faces/triangles. Since all $n$ non-diagonal edges of $T_C$ (the sides of the polygon) are sides of the $n-2$ triangles, by the pigeonhole principle, there are at least 2 triangles which use 2 sides of the $n$-gon as their sides. Then, those triangles' third edge must be an span-2 one and, as soon as $n\geq 5$, those third edges do not coincide.
\end{proof}

\begin{proposition}
A convex geometric graph with $n\geq 5$ and only two span-2 edges which appear \emph{consecutive} in the drawing has no triangulation. Equivalently, $T(K_{n,-(n-2)})=0$.
\label{pro:nminus2}
\end{proposition}

\begin{proof} 
Two consecutive span-2 edges intersect, therefore together they do not form any triangulation, unless only one is needed (quadrilateral). So there must be a third span-2 edge which, together with one of the 2 consecutive ones, is the obligatory second span-2 edge in a triangulation $T$ of the graph (see Proposition \ref{pro:span22}). Contradiction, as the graph has no other such edge. 
\end{proof}

\subsection{The reduction argument}

Our reduction argument is primarily based on partitioning the number of triangulations of some convex graph $G$ into those which include a specific span-2 edge $e_i$ and those which do not (assuming there is at least one span-2 edge in $G$). Let us already point out the tree structure implied by the partitioning.

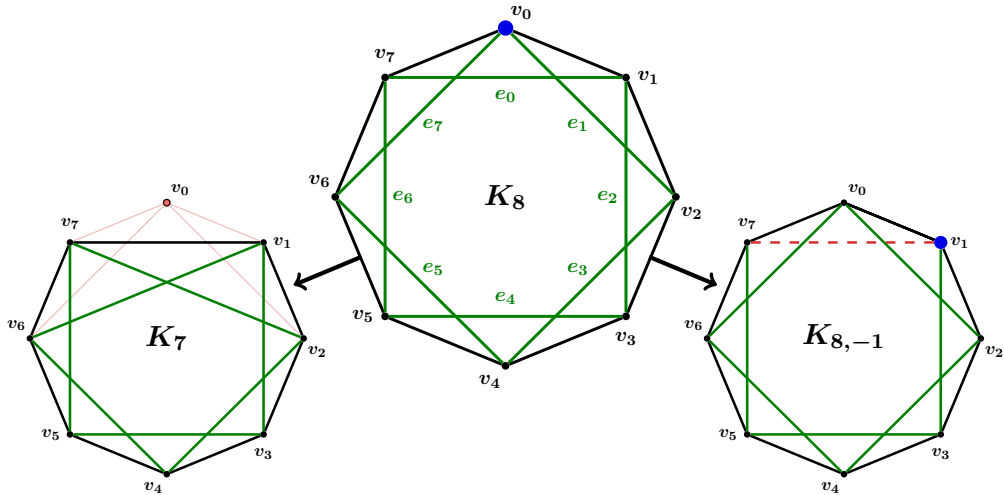
\begin{figure}[h!]
\centering
\begin{tikzpicture}
   [scale=0.9375, every node/.style={rectangle, draw=none, fill=none, scale=0.2}]
  
	\node (k8) [scale=7] at (0,1) {\begin{tikzpicture}
   [scale=0.8, every node/.style={circle, fill=black, scale=0.2}]
      \node [fill=keyb, scale=2]  (n0) at (90:2)  {};
  \node (n1) at (45:2) {};
  \node (n2) at (0:2) {};
  \node (n3) at (-45:2) {};
   \node (n4) at (-90:2) {};
  \node (n5) at (-135:2)  {};
   \node (n6) at (180:2) {};
  \node (n7) at (135:2)  {};
   \draw [black, thick] (n0)--(n1)--(n2)--(n3)--(n4)--(n5)--(n6)--(n7)--(n0);
   
     \foreach \x/\y in {n0/n2,n1/n3,n2/n4,n3/n5,n4/n6,n5/n7,n6/n0,n7/n1}
           \draw [keyg,thick] (\x)--(\y);    
        
        \node  [scale=2.5, fill=none, draw=none, keyg] at (90:1.2) {$\bm{e_0}$};
     \node [scale=2.5, fill=none, draw=none, keyg] at (45:1.2) {$\bm{e_1}$};
     \node [scale=2.5, fill=none, draw=none, keyg] at (0:1.2) {$\bm{e_2}$};
     \node [scale=2.5, fill=none, draw=none, keyg] at (-45:1.2) {$\bm{e_3}$};
     \node [scale=2.5, fill=none, draw=none, keyg] at (-90:1.2) {$\bm{e_4}$};
     \node [scale=2.5, fill=none, draw=none, keyg] at (-135:1.2) {$\bm{e_5}$};
     \node [scale=2.5, fill=none, draw=none, keyg] at (180:1.2) {$\bm{e_6}$};
     \node [scale=2.5, fill=none, draw=none, keyg] at (135:1.2) {$\bm{e_7}$}; 
        \node  [scale=2.5, fill=none, draw=none, anchor=south west] at (n0.center) {$\bm{v_0}$};
     \node [scale=2.5, fill=none, draw=none, anchor=west] at (n1.center) {$\bm{v_1}$};
     \node [scale=2.5, fill=none, draw=none,  anchor=north west] at (n2.center) {$\bm{v_2}$};
     \node [scale=2.5, fill=none, draw=none,  anchor=north] at (n3.center) {$\bm{v_3}$};
     \node [scale=2.5, fill=none, draw=none,  anchor=north east] at (n4.center) {$\bm{v_4}$};
     \node [scale=2.5, fill=none, draw=none,  anchor=east] at (n5.center) {$\bm{v_5}$};
     \node [scale=2.5, fill=none, draw=none,  anchor=south east] at (n6.center) {$\bm{v_6}$};
     \node [scale=2.5, fill=none, draw=none,  anchor=south] at (n7.center) {$\bm{v_7}$};
     
 \end{tikzpicture}};
     \node  [scale=5, fill=none, draw=none] at (0,1) {$\bm{K_8$}};
     
     \node (k87) [scale=6] at (-4.75,-1) {\begin{tikzpicture}
   [scale=0.75, every node/.style={circle, draw=none, fill=black, scale=0.2}]
      \node [draw, fill=red!60] (n0) at (90:2)  {};
  \node (n1) at (45:2) {};
  \node (n2) at (0:2) {};
  \node (n3) at (-45:2) {};
   \node (n4) at (-90:2) {};
  \node (n5) at (-135:2)  {};
   \node (n6) at (180:2) {};
  \node (n7) at (135:2)  {};
  \foreach \x/\y in {n0/n2,n6/n0,n7/n0,n0/n1}
           \draw [keyr!50,opacity=0.5] (\x)--(\y);    
           
   \draw [black, thick] (n1)--(n2)--(n3)--(n4)--(n5)--(n6)--(n7)--(n1);
   
         \foreach \x/\y in {n1/n3,n2/n4,n3/n5,n4/n6,n5/n7,n6/n1,n7/n2}
           \draw [keyg,thick] (\x)--(\y);    
                      
    \node  [scale=2.5, fill=none, draw=none, anchor=south west] at (n0.center) {$\bm{v_0}$};
     \node [scale=2.5, fill=none, draw=none, anchor=west] at (n1.center) {$\bm{v_1}$};
     \node [scale=2.5, fill=none, draw=none,  anchor=north west] at (n2.center) {$\bm{v_2}$};
     \node [scale=2.5, fill=none, draw=none,  anchor=north] at (n3.center) {$\bm{v_3}$};
     \node [scale=2.5, fill=none, draw=none,  anchor=north east] at (n4.center) {$\bm{v_4}$};
     \node [scale=2.5, fill=none, draw=none,  anchor=east] at (n5.center) {$\bm{v_5}$};
     \node [scale=2.5, fill=none, draw=none,  anchor=south east] at (n6.center) {$\bm{v_6}$};
     \node [scale=2.5, fill=none, draw=none,  anchor=south] at (n7.center) {$\bm{v_7}$};

    \end{tikzpicture}};
        \node  [scale=5, fill=none, draw=none] at (-4.75,-1) {$\bm{K_7$}};
       
        	\node (k881) [scale=6] at (4.75,-1) {\begin{tikzpicture}
   [scale=0.75, every node/.style={circle, draw=none, fill=black, scale=0.2}]
      \node (n0) at (90:2)  {};
  \node [fill=keyb, scale=2] (n1) at (45:2) {};
  \node (n2) at (0:2) {};
  \node (n3) at (-45:2) {};
   \node (n4) at (-90:2) {};
  \node (n5) at (-135:2)  {};
   \node (n6) at (180:2) {};
  \node (n7) at (135:2)  {};
   \draw [black, thick] (n0)--(n1)--(n2)--(n3)--(n4)--(n5)--(n6)--(n7)--(n0)--(n1);
   
    \foreach \x/\y in {n1/n7}
           \draw [keyr!80, thick, dashed] (\x)--(\y);    

     \foreach \x/\y in {n1/n3,n2/n4,n3/n5,n4/n6,n5/n7,n6/n0,n0/n2}
           \draw [keyg,thick] (\x)--(\y);    
                        
    \node  [scale=2.5, fill=none, draw=none, anchor=south west] at (n0.center) {$\bm{v_0}$};
     \node [scale=2.5, fill=none, draw=none, anchor=west] at (n1.center) {$\bm{v_1}$};
     \node [scale=2.5, fill=none, draw=none,  anchor=north west] at (n2.center) {$\bm{v_2}$};
     \node [scale=2.5, fill=none, draw=none,  anchor=north] at (n3.center) {$\bm{v_3}$};
     \node [scale=2.5, fill=none, draw=none,  anchor=north east] at (n4.center) {$\bm{v_4}$};
     \node [scale=2.5, fill=none, draw=none,  anchor=east] at (n5.center) {$\bm{v_5}$};
     \node [scale=2.5, fill=none, draw=none,  anchor=south east] at (n6.center) {$\bm{v_6}$};
     \node [scale=2.5, fill=none, draw=none,  anchor=south] at (n7.center) {$\bm{v_7}$};

    \end{tikzpicture}};
        \node  [scale=5, fill=none, draw=none] at (4.75,-1) {$\bm{K_{8,-1}$}};

\draw [ultra thick, ->] ($(k8.center)!0.425!(k87.center)$)--($(k8.center)!0.625!(k87.center)$);

\draw [ultra thick, ->] ($(k8.center)!0.425!(k881.center)$)--($(k8.center)!0.625!(k881.center)$);

    \end{tikzpicture}
 \caption{Reduction for the complete $K_n$, working on $v_0$/$e_0$}
 \label{fig:kn0red}
\end{figure}

\begin{proposition}
Let $G=(V,E)$ be a convex geometric graph and $T(G)$ the number of its triangulations. Then, for any span-2 edge $e_i$ we have that T($G \,|\, e_i \text{ included}$) = T($G \setminus v_i$). 

Obviously, the following holds: $T(G) = T(G \setminus v_i) + T(G \setminus e_i)$.

\label{pro:selfred}
\end{proposition}

\begin{proof}
Observe that if $e_i$ is selected to participate in a triangulation $T_C$, then all edges incident to $v_i$ may not participate in $T_C$, as they all cross $e_i$. Therefore, the number of triangulations where $e_i$ is included is exactly determined by the number of triangulations of $G\setminus v_i$, as adding the ``hat'' $v_{i-1}v_iv_{i+1}$ in any of the latest will yield a distinct triangulation of the original graph. \end{proof}

Figure \ref{fig:kn0red} demonstrates the above, nevertheless for the special case of an initial complete graph. However, we easily get that $T(K_{n,0})=T(K_{n-1,0})+T(K_{n,-1})$, which is essentially the first of the BC relations. 

Now, consider a $K_{n,-1}$ graph. If we work on $v_1$/$e_1$ (Figure \ref{fig:kn1red}, marked blue vertex), thus the selection or not of $e_1$ for the next branching of our tree, we obtain either a $K_7\equiv K_{7,0}$ (left child), either a $K_{8,-2}$. In general (Figure \ref{fig:knmred}), working on the \emph{next available} span-2 edge, we reduce $K_{n,-m}$ is to either a $K_{n-1,-m+1}$ (left child) or a $K_{n,-m-1}$ (right child).

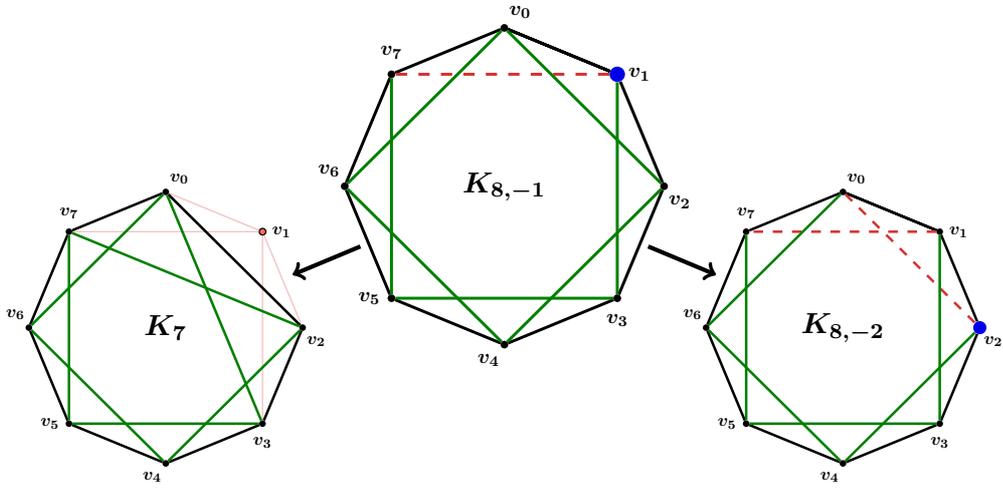
\begin{figure}[h]
\centering
\begin{tikzpicture}
   [scale=0.9375, every node/.style={rectangle, draw=none, fill=none, scale=0.2}]
  
	\node (k8) [scale=7] at (0,1) {\begin{tikzpicture}
   [scale=0.75, every node/.style={circle, draw=none, fill=black, scale=0.2}]
      \node (n0) at (90:2)  {};
  \node [fill=keyb, scale=2] (n1) at (45:2) {};
  \node (n2) at (0:2) {};
  \node (n3) at (-45:2) {};
   \node (n4) at (-90:2) {};
  \node (n5) at (-135:2)  {};
   \node (n6) at (180:2) {};
  \node (n7) at (135:2)  {};
   \draw [black, thick] (n0)--(n1)--(n2)--(n3)--(n4)--(n5)--(n6)--(n7)--(n0)--(n1);
   
    \foreach \x/\y in {n1/n7}
           \draw [keyr!80, thick, dashed] (\x)--(\y);    

     \foreach \x/\y in {n1/n3,n2/n4,n3/n5,n4/n6,n5/n7,n6/n0,n0/n2}
           \draw [keyg,thick] (\x)--(\y);    
                        
    \node  [scale=2.5, fill=none, draw=none, anchor=south west] at (n0.center) {$\bm{v_0}$};
     \node [scale=2.5, fill=none, draw=none, anchor=west] at (n1.center) {$\bm{v_1}$};
     \node [scale=2.5, fill=none, draw=none,  anchor=north west] at (n2.center) {$\bm{v_2}$};
     \node [scale=2.5, fill=none, draw=none,  anchor=north] at (n3.center) {$\bm{v_3}$};
     \node [scale=2.5, fill=none, draw=none,  anchor=north east] at (n4.center) {$\bm{v_4}$};
     \node [scale=2.5, fill=none, draw=none,  anchor=east] at (n5.center) {$\bm{v_5}$};
     \node [scale=2.5, fill=none, draw=none,  anchor=south east] at (n6.center) {$\bm{v_6}$};
     \node [scale=2.5, fill=none, draw=none,  anchor=south] at (n7.center) {$\bm{v_7}$};

    \end{tikzpicture}};
     \node  [scale=5, fill=none, draw=none] at (0,1) {$\bm{K_{8,-1}$}};
     
     \node (k87) [scale=6] at (-4.75,-1) {\begin{tikzpicture}
   [scale=0.75, every node/.style={circle, draw=none, fill=black, scale=0.2}]
      \node  (n0) at (90:2)  {};
  \node [draw, fill=red!60] (n1) at (45:2) {};
  \node (n2) at (0:2) {};
  \node (n3) at (-45:2) {};
   \node (n4) at (-90:2) {};
  \node (n5) at (-135:2)  {};
   \node (n6) at (180:2) {};
  \node (n7) at (135:2)  {};

    \foreach \x/\y in {n1/n7,n1/n2,n1/n0,n1/n3}
           \draw [keyr!50, opacity=0.5] (\x)--(\y);    
           
   \draw [black, thick] (n0)--(n2)--(n3)--(n4)--(n5)--(n6)--(n7)--(n0);

     \foreach \x/\y in {n2/n4,n3/n5,n4/n6,n5/n7,n6/n0,n0/n3,n2/n7}
           \draw [keyg,thick] (\x)--(\y);    
                        
    \node  [scale=2.5, fill=none, draw=none, anchor=south west] at (n0.center) {$\bm{v_0}$};
     \node [scale=2.5, fill=none, draw=none, anchor=west] at (n1.center) {$\bm{v_1}$};
     \node [scale=2.5, fill=none, draw=none,  anchor=north west] at (n2.center) {$\bm{v_2}$};
     \node [scale=2.5, fill=none, draw=none,  anchor=north] at (n3.center) {$\bm{v_3}$};
     \node [scale=2.5, fill=none, draw=none,  anchor=north east] at (n4.center) {$\bm{v_4}$};
     \node [scale=2.5, fill=none, draw=none,  anchor=east] at (n5.center) {$\bm{v_5}$};
     \node [scale=2.5, fill=none, draw=none,  anchor=south east] at (n6.center) {$\bm{v_6}$};
     \node [scale=2.5, fill=none, draw=none,  anchor=south] at (n7.center) {$\bm{v_7}$};

    \end{tikzpicture}};
        \node  [scale=5, fill=none, draw=none] at (-4.75,-1) {$\bm{K_7}$};
       
        	\node (k881) [scale=6] at (4.75,-1) {\begin{tikzpicture}
   [scale=0.75, every node/.style={circle, draw=none, fill=black, scale=0.2}]
      \node (n0) at (90:2)  {};
  \node (n1) at (45:2) {};
  \node [fill=keyb, scale=2] (n2) at (0:2) {};
  \node (n3) at (-45:2) {};
   \node (n4) at (-90:2) {};
  \node (n5) at (-135:2)  {};
   \node (n6) at (180:2) {};
  \node (n7) at (135:2)  {};
   \draw [black, thick] (n0)--(n1)--(n2)--(n3)--(n4)--(n5)--(n6)--(n7)--(n0)--(n1);
   
    \foreach \x/\y in {n1/n7,n0/n2}
           \draw [keyr!80, thick, dashed] (\x)--(\y);    

     \foreach \x/\y in {n1/n3,n2/n4,n3/n5,n4/n6,n5/n7,n6/n0}
           \draw [keyg,thick] (\x)--(\y);    
                        
    \node  [scale=2.5, fill=none, draw=none, anchor=south west] at (n0.center) {$\bm{v_0}$};
     \node [scale=2.5, fill=none, draw=none, anchor=west] at (n1.center) {$\bm{v_1}$};
     \node [scale=2.5, fill=none, draw=none,  anchor=north west] at (n2.center) {$\bm{v_2}$};
     \node [scale=2.5, fill=none, draw=none,  anchor=north] at (n3.center) {$\bm{v_3}$};
     \node [scale=2.5, fill=none, draw=none,  anchor=north east] at (n4.center) {$\bm{v_4}$};
     \node [scale=2.5, fill=none, draw=none,  anchor=east] at (n5.center) {$\bm{v_5}$};
     \node [scale=2.5, fill=none, draw=none,  anchor=south east] at (n6.center) {$\bm{v_6}$};
     \node [scale=2.5, fill=none, draw=none,  anchor=south] at (n7.center) {$\bm{v_7}$};

    \end{tikzpicture}};
        \node  [scale=5, fill=none, draw=none] at (4.75,-1) {$\bm{K_{8,-2}$}};

\draw [ultra thick, ->] ($(k8.center)!0.425!(k87.center)$)--($(k8.center)!0.625!(k87.center)$);

\draw [ultra thick, ->] ($(k8.center)!0.425!(k881.center)$)--($(k8.center)!0.625!(k881.center)$);

    \end{tikzpicture}
 \caption{Reduction for $K_{n,-1}$, working on the \emph{next} vertex/span-2 edge}
 \label{fig:kn1red}
\end{figure}

\begin{figure}[h]
\centering
\begin{tikzpicture}
   [scale=0.9375, every node/.style={rectangle, draw=none, fill=none, scale=0.2}]
  
	\node (k8) [scale=7] at (0,1) {\begin{tikzpicture}
   [scale=0.75, every node/.style={circle, draw=none, fill=black, scale=0.2}]
      \node (n0) at (90:2)  {};
  \node (n1) at (45:2) {};
  \node  (n2) at (0:2) {};
  \node [fill=keyb, scale=2] (n3) at (-45:2) {};
   \node (n4) at (-90:2) {};
  \node (n5) at (-135:2)  {};
   \node (n6) at (180:2) {};
  \node (n7) at (135:2)  {};
   \draw [black, thick] (n0)--(n1)--(n2)--(n3)--(n4)--(n5)--(n6)--(n7)--(n0)--(n1);
   
    \foreach \x/\y in {n1/n3, n1/n7,n0/n2}
           \draw [keyr!80, thick, dashed] (\x)--(\y);    

     \foreach \x/\y in {n2/n4,n3/n5,n4/n6,n5/n7,n6/n0}
           \draw [keyg,thick] (\x)--(\y);    
                        
    \node  [scale=2.5, fill=none, draw=none, anchor=south west] at (n0.center) {$\bm{v_0}$};
     \node [scale=2.5, fill=none, draw=none, anchor=west] at (n1.center) {$\bm{v_1}$};
     \node [scale=2.5, fill=none, draw=none,  anchor=north west] at (n2.center) {$\bm{v_2}$};
     \node [scale=2.5, fill=none, draw=none,  anchor=north] at (n3.center) {$\bm{v_3}$};
     \node [scale=2.5, fill=none, draw=none,  anchor=north east] at (n4.center) {$\bm{v_4}$};
     \node [scale=2.5, fill=none, draw=none,  anchor=east] at (n5.center) {$\bm{v_5}$};
     \node [scale=2.5, fill=none, draw=none,  anchor=south east] at (n6.center) {$\bm{v_6}$};
     \node [scale=2.5, fill=none, draw=none,  anchor=south] at (n7.center) {$\bm{v_7}$};

    \end{tikzpicture}};
     \node  [scale=5, fill=none, draw=none] at (0,1) {$\bm{K_{8,-3}$}};
     
     \node (k87) [scale=6] at (-4.75,-1) {\begin{tikzpicture}
   [scale=0.75, every node/.style={circle, draw=none, fill=black, scale=0.2}]
      \node (n0) at (90:2)  {};
  \node (n1) at (45:2) {};
  \node [fill=keyb, scale=2]   (n2) at (0:2) {};
  \node [draw, fill=red!60]  (n3) at (-45:2) {};
   \node  (n4) at (-90:2) {};
  \node (n5) at (-135:2)  {};
   \node (n6) at (180:2) {};
  \node (n7) at (135:2)  {};
   
    \foreach \x/\y in {n1/n3,n3/n4,n3/n2,n3/n5}
           \draw [keyr!50, opacity=0.5] (\x)--(\y);    
   
   \draw [black, thick] (n0)--(n1)--(n2)--(n4)--(n5)--(n6)--(n7)--(n0)--(n1);
   
    \foreach \x/\y in {n0/n2,n1/n7}
           \draw [keyr!80, thick, dashed] (\x)--(\y);    

     \foreach \x/\y in {n1/n4,n4/n6,n5/n7,n6/n0,n2/n5}
           \draw [keyg,thick] (\x)--(\y);    
                        
    \node  [scale=2.5, rectangle, fill=none, draw=none, anchor=south west] at (n0.center) {$\bm{v_0 \equiv r_0}$};
     \node [scale=2.5, rectangle, fill=none, draw=none, anchor=west] at (n1.center) {$\bm{v_1 \equiv r_1}$};
     \node [scale=2.5, rectangle, fill=none, draw=none,  anchor=north west] at (n2.center) {$\bm{v_2 \equiv r_2}$};
     \node [scale=2.5, rectangle, fill=none, draw=none,  anchor=north west ] at (n3.center) {$\bm{v_3}$};
     \node [scale=2.5, rectangle, fill=none, draw=none,  anchor=north east] at (n4.center) {$\bm{v_4 \equiv r_3}$};
     \node [scale=2.5, rectangle, fill=none, draw=none,  anchor=east] at (n5.center) {$\bm{v_5 \equiv r_4}$};
     \node [scale=2.5, rectangle, fill=none, draw=none,  anchor=south east] at (n6.center) {$\bm{v_6 \equiv r_5}$};
     \node [scale=2.5, rectangle, fill=none, draw=none,  anchor=south east ] at (n7.center) {$\bm{v_7 \equiv r_6}$};

    \end{tikzpicture}};
        \node  [scale=5, fill=none, draw=none] at (-4.75,-1) {$\bm{K_{7,-2}}$};
       
        	\node (k881) [scale=6] at (4.75,-1) {\begin{tikzpicture}
   [scale=0.75, every node/.style={circle, draw=none, fill=black, scale=0.2}]
      \node (n0) at (90:2)  {};
  \node (n1) at (45:2) {};
  \node (n2) at (0:2) {};
  \node (n3) at (-45:2) {};
   \node [fill=keyb, scale=2] (n4) at (-90:2) {};
  \node  (n5) at (-135:2)  {};
   \node (n6) at (180:2) {};
  \node (n7) at (135:2)  {};
   \draw [black, thick] (n0)--(n1)--(n2)--(n3)--(n4)--(n5)--(n6)--(n7)--(n0)--(n1);
   
    \foreach \x/\y in {n2/n4,n1/n3, n1/n7,n0/n2}
           \draw [keyr!80, thick, dashed] (\x)--(\y);    

     \foreach \x/\y in {n3/n5,n4/n6,n5/n7,n6/n0}
           \draw [keyg,thick] (\x)--(\y);    
                        
    \node  [scale=2.5, fill=none, draw=none, anchor=south west] at (n0.center) {$\bm{v_0}$};
     \node [scale=2.5, fill=none, draw=none, anchor=west] at (n1.center) {$\bm{v_1}$};
     \node [scale=2.5, fill=none, draw=none,  anchor=north west] at (n2.center) {$\bm{v_2}$};
     \node [scale=2.5, fill=none, draw=none,  anchor=north] at (n3.center) {$\bm{v_3}$};
     \node [scale=2.5, fill=none, draw=none,  anchor=north east] at (n4.center) {$\bm{v_4}$};
     \node [scale=2.5, fill=none, draw=none,  anchor=east] at (n5.center) {$\bm{v_5}$};
     \node [scale=2.5, fill=none, draw=none,  anchor=south east] at (n6.center) {$\bm{v_6}$};
     \node [scale=2.5, fill=none, draw=none,  anchor=south] at (n7.center) {$\bm{v_7}$};

    \end{tikzpicture}};
        \node  [scale=5, fill=none, draw=none] at (4.75,-1) {$\bm{K_{8,-4}$}};

\draw [ultra thick, ->] ($(k8.center)!0.425!(k87.center)$)--($(k8.center)!0.625!(k87.center)$);

\draw [ultra thick, ->] ($(k8.center)!0.425!(k881.center)$)--($(k8.center)!0.625!(k881.center)$);

    \end{tikzpicture}
 \caption{Reduction for $K_{n,-m}$, working on the \emph{next} vertex/span-2 edge}
 \label{fig:knmred}
\end{figure}
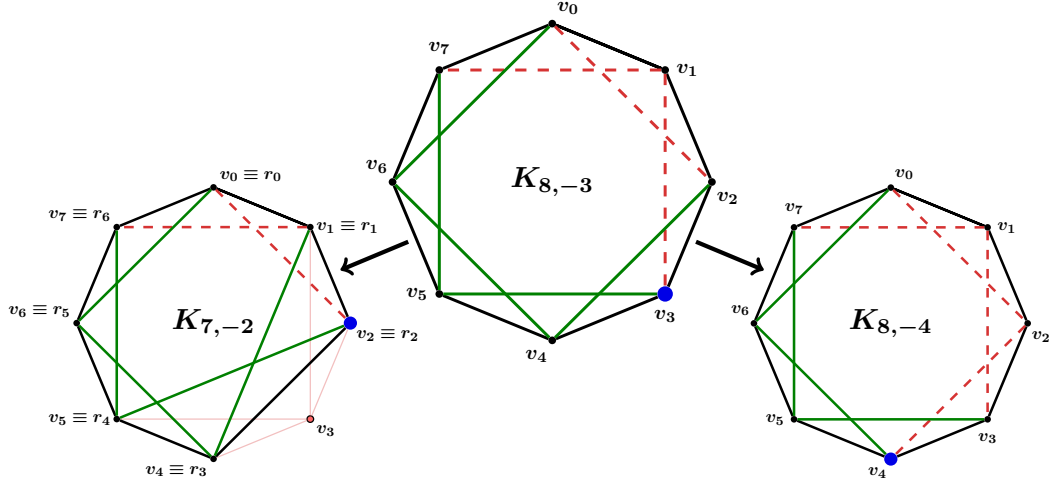

It proves that the key to our reduction is the correct definition of the \emph{next vertex (or span-2 edge)} to work on. If, for instance, while on $K_{8,-1}$ we select to work on $v_5$/$e_5$, the nice to our reduction $K_7$ does not pop out. More importantly, we would miss the recurring of \emph{isomorphic graphs}, which give us the number of triangulations they feature per class ($K_{n,-m}$), and not per specific instance.

\begin{definition}[Next vertex/span-2 edge]
Given a properly labeled convex geometric graph on $n$ vertices, missing only $m$ consecutive span-2 edges, $e_i,e_{i+1},...,e_{i+m}$. The \emph{next} vertex/span-2 edge is the pair $v_{m+1}$/$e_{m+1}$.
\label{def:next}
\end{definition}  
 
 Now, let us formally prove our claim.
 
\begin{theorem}
For any $K_{n,-m}$ and $-n+1 \leq -m \leq -1$, it is $$T(K_{n,-m})=T(K_{n-1,-m+1})+T(K_{n,-m-1})$$
Also, it is $$T(K_{n,-m}) = a_{n-2,m}$$ for all defined quantities.
\end{theorem}

\begin{proof}
Without loss of generality (see \ref{rul:relabel}), we may relabel our initial graph --if necessary-- and have the edges $e_0,...,e_{m-1}$ as the missing span-2 ones. The next vertex/span-2 edge is the pair $v_m$/$e_m$. The right subtree is rooted to a $K_{n,-m-1}$ graph, as it derives from deleting $e_m$, having a consecutive $m+1$ span-2 (but no other) edges missing. 

For the left subtree, deleting $v_m$ leaves us to examine the edges of a graph on $n-1$ vertices. We have that edges $v_{n-1}v_1,v_0v_2,...,v_{m-3}v_{m-1}$ are all missing from the induced subgraph, and together they add up to $m-1$ consecutive span-2 edges. Edge $v_{m-2}v_m$, missing from $K_{n,-m}$ is not a potential edge of the smaller graph. Edge $v_{m-1}v_{m+1}$ has now become a side of the $(n-1)$-gon, and $v_{m-2}v_{m+1}$ and $v_{m-1}v_{m+2}$ are now span-2 edges of the new graph, as are all the (potentially) remaining $v_{m+1}v_{m+3},...,v_{n-2}v_0$. Thus, we have obtained a $K_{n-1,-m+1}$ graph. To complete the proof, we remind that:
\begin{itemize}
\item Proposition \ref{pro:nminus2} gives us $T(K_{n,-n+2})=0$ for pentagons and up; but $T(K_{4,-2})$ gives the number of triangulations of a quadrilateral missing all 2 span-2 edges, therefore the relation is satisfied for all $n\geq 4$.  
\item $T(K_{3,0})=1$, as $K_3$ is a triangle.
\item $K_{3,-1}$ and $K_{2,0}$ are not properly defined as convex graphs, and not actually needed, but for sake of completeness, we shall consider their number of triangulations equal to $a_{1,1}$ and $a_{0,0}$ respectively.
\end{itemize}
\end{proof}

\begin{theorem}
Our main sampling algorithm scheme adopted for convex triangulations yields an algorithm of running time $\mathcal{O}(n^2)$.
\end{theorem}

\begin{proof} As an outline, the quadratic (and not linear) time is due to the information that must be stored when branching left. For a complete proof, refer to the Appendix. Yet, we note that \cite{dev99ran} implies a linear algorithm by sampling monotone lattice paths, first. However, if one wants to preserve one-to-one relation between the subinstances of the two problems, we are not aware of any work proposing any similar framework to ours. 
\end{proof}

\section{Mountain ranges}

The mountain range problem consider the ways to form a ``mountain range'' with $n$ upstrokes and $n$ downstrokes that all stay above a horizontal line. As with the convex triangulations, we assume the space of all solutions, that is actually a $n \times n$ lattice rotated as to comprehend our node to node movement as Up or Down and always maintaining direction to the right (w.l.o.g.). 

Assume we are on the horizon, or level 0, so we \emph{have} to move Up. Our actual choice is whether to go further Up or slant Down. The latest case brings us to the horizon with 1 upstroke and downstroke to our left, therefore the remaining mountain ranges are exactly the mountain ranges with $n-1$ upstrokes and downstrokes - correlate to the left child of the BC-implied tree. That is:
$$r_{n,0}=r_{n-1,0}+r_{n,1}$$ where the remainder $r_{n,1}$ denotes the case we moved Up. Of course, the cases yield mutually disjoint mountain ranges.

It is easy enough to go on, explaining that if we are not on the horizon ($m\geq 1$), we are not dictated to move up (unless we cannot - the ``height'' or the range is at most $n$). If we move up, our $m$ increases by one. If we move down, we ``see'' again a disjoint set of potential solutions. Since we slanted down ($m$ decreased) we can symbolize this reduction by $$r_{n,m}=r_{n-1,m-1}+r_{n,m+1}$$.

\begin{theorem}
The main algorithm obtained by our scheme for the mountain ranges problem is linear.
\end{theorem}

\begin{proof}
We need linear in $n$ coin tosses to get the code of the mountain range; we branch linear in $n$ times and for each node we lie on, we store a linear amount of information (Ups-Downs) to backtrack to the root.  
\end{proof}


\appendix

\section{Omitted proofs}
\label{sec:appa}

\begin{proof}[Proof of Theorem \ref{thm:bal}]
We will consider the 2 main of the BC-relations. The marginal cases are trivial to prove.
\begin{align*}
a_{n,0}&=a_{n-1,0}+a_{n,1} &\Leftrightarrow \\[1mm] 
N_{n,n-1}&=N_{n-1,n-2}+N_{n,n-2} &\Leftrightarrow \\[1mm] 
\frac{1}{n}{2n \choose n-1}&= \frac{1}{n-1}{2n-2 \choose n-2}+\frac{3}{2n-1}{2n-1 \choose n-2} &\Leftrightarrow \\[1mm] 
\frac{(2n)!}{n!(n+1)!}&= \frac{(2n-2)!}{(n-1)!n!}+\frac{3(2n-2)!}{(n-2)!(n+1)!} &\Leftrightarrow \\[1mm] 
\frac{(2n-1)(2n)}{(n+1)!}&= \frac{1}{(n-1)!}+\frac{3}{(n-2)!(n+1)} &\Leftrightarrow \\[1mm] 
\frac{2(2n-1)}{(n-1)(n+1)}&= \frac{1}{(n-1)}+\frac{3}{(n+1)} &\Leftrightarrow \\[1mm]
4n-2 &= (n+1)+3(n-1), \quad \quad \text{which holds.}
\end{align*}
Regarding $a_{n,m}=a_{n-1,m-1}+a_{n,m+1}$ we have:
\begin{align*}
N_{n,m}&=N_{n-1,m}+N_{n,m-1} &\Leftrightarrow \\[1mm] 
\frac{n+1-m}{n+1+m}{n+1+m \choose m}&= \frac{n-m}{n+m}{n+m \choose m}+\frac{n-m+2}{n+m}{n+m \choose m-1} &\Leftrightarrow \\[1mm] 
\frac{(n+1-m)(n+m)!}{m!(n+1)!}&= \frac{(n-m)(n+m)!}{(n+m)m!n!}+\frac{(n-m+2)(n+m)!}{(n+m)(m-1)!(n+1)!} &\Leftrightarrow \\[1mm]
\frac{(n+1-m)}{m(n+1)}&= \frac{(n-m)}{(n+m)m}+\frac{(n-m+2)}{(n+m)(n+1)} &\Leftrightarrow \\[1mm]
\frac{1}{m}-\frac{1}{n+1}&= \frac{1}{m}-\frac{2}{n+m}+\frac{(n-m+2)}{(n+m)(n+1)} &\Leftrightarrow \\[1mm]
\frac{2}{n+m}-\frac{1}{n+1}&=\frac{(n-m+2)}{(n+m)(n+1)} &\Leftrightarrow \\[1mm]
2(n+1)-(n+m)&=n-m+2, \quad \quad \text{which holds.}
\end{align*}
Therefore, we have completed our proof.
\end{proof}

\begin{proof}[Proof of Theorem \ref{thm:gen}]
Let $G$ be the generating function for the numbers $a_{i,j}$, i.e.
\[G(x,y)= \sum_{i=0}^{\infty} \sum_{j=0}^{\infty} a_{i,j} x^i y^j\] and let $B$ be the generating function for numbers $b_{i,j}=a_{i+1,j}$ for all $i,j\geq 0.$ We will first determine $B$, and then it is simply:
\begin{equation}
\label{6}
G=xB+1.
\end{equation}
Initially, let us break $B$ into sums of fixed $j$:
\begin{equation}
B=A_0+A_1+A_2+...= \sum_{k=0}^{\infty} A_k, \label{1}
\end{equation}
where for all $k\in \mathbb{N}$ $A_k=\sum_{i=0}^{\infty} b_{i,k} x^i y^k$. From the BC-relation \ref{rel:bc1} we obtain:
\begin{equation}\label{2}
A_0=\frac{1}{y} A_1 + x A_0 + 1,
\end{equation}
while from the recursive BC-relation \ref{rel:bc2} we obtain:
\begin{equation}\label{3}
\forall k>0 \ A_k=\frac{1}{y} A_{k+1} + xyA_{k-1}.
\end{equation}
We have also showed that the numbers $a_{i,0}$ are actually Catalan numbers, in particular $b_{i,0}=a_{i+1,0}=C_{i+1},\text{ for }  i\geq 0$. So, if $Cat$ is the generating function of the Catalan numbers, then:
\begin{equation}\label{4}
Cat=x A_0 +1 \Rightarrow A_0=\frac{Cat-1}{x}.
\end{equation}
It is well known that: 
\begin{equation}\label{5}
Cat(x)=\frac{1 - \sqrt{1 - 4 x}}{2 x}.
\end{equation}
In order to get the generating function $G$, plug equations \ref{2},\ref{3},\ref{4},\ref{5} into equation \ref{1}:
\begin{align*}
B&=\frac{1}{y} (A_1+A_2+...)+ xy (A_0+ A_1 +A_2 +...) + x A_0 +1 \\
&=\frac{1}{y}(G-A_0)+ xy G + x A_0 +1 \\
&=\frac{y+ (xy-1) A_0} {y-1-xy^2} =\frac{y+ (xy-1)\frac{Cat(x)-1}{x}}{y-1-xy^2}\\
&=\frac{y}{y-1-xy^2}+\frac{(xy-1)(1-\sqrt{1-4x}-2x)}{2x^2(y-1-xy^2)}.
\end{align*}
Finally, from relation \ref{6} he have:
\[G(x,y)=xB(x,y)+1=\frac{(1-xy) (1-\sqrt{1-4x}-2xy)} {2x(1-y+xy^2)}.\]
\end{proof}

Observe that we did not directly use the initial condition that $a_{n,n}=0, n\geq 1$; we actually did account for it, as it is hidden in the proof that the numbers $a_{n,0}$ correspond to Catalan numbers.

\section{Notes on the main sampling algorithm for convex triangulations}

For our convenience, we hereby adopt the notation $T_{n,-m}=T(K_{n,-m})$.


The left child of a $K_{n,-m}$ in our reduction is a graph on $n-1$ vertices, in other words, w.r.t.\ our initial proper vertex labeling, the labeling of the left child is not proper, as a vertex is missing. Therefore, we must \emph{relabel} this graph, and every left child graph of the tree. The following Rule gives the final ingredient for building the desirable reduction tree.

\begin{rul}
Let $K_{r,-m}$ be the left child of a node of a tree rooted at $K_n$. If $m=0$, then relabel the complete $K_r$ as desired. Else, place $r_0$ opposite to the first missing span-2 edge and complete with a proper relabeling. For each new graph, maintain a $1\times n$ table to indicate the mapping of the current to the initial labeling; current missing vertices of the parent graph can be marked (e.g.\ $-1$ entries). 
\label{rul:relabel}
\end{rul} 

\subsubsection{Properties of the reduction tree} 
\label{sec:treeprop}
In all, we have build a triangulation tree with the following properties:
\begin{itemize}
\item Every node has 2 children (left and right -- binary tree).
\item Every node indicating $K_{r,-m}$ needs $\mathcal{O}(n)$ time to be created and coded: store integers $r$ and $m$, the $1\times n$ mapping matrix for the vertex labeling, plus a 2 integer indicator for backtracking to the parent node: if it is a left child, hold $[v_m,-1]$, for a right child hold $[v_{m-1},v_{m+1}]$ where $v_m$ is the critical vertex of the parent node w.r.t.\ which the node branched (store as labeled in the parent node).  
\item For any node, the triangulations coded in the left subtree are disjoint to the ones of the right subtree, as they differ at the edge w.r.t.\ which the node branched. 
\item If a node is $K_3$, then stop and mark \emph{one triangulation} (green leaf). Due to the above, all green leaves are left children and there is a total of exactly $\bm{C_{n-2}}$ green leaves in the tree.
\item If a node is of the form $K_{r,-r+2}$, then stop and mark \emph{no triangulation} (red leaf)
\item Considering all the above, from any green leaf, backtracking to the root is equivalent to obtaining one specific triangulation; this is achieved in $\mathcal{O}(n)$ time. 
\end{itemize}

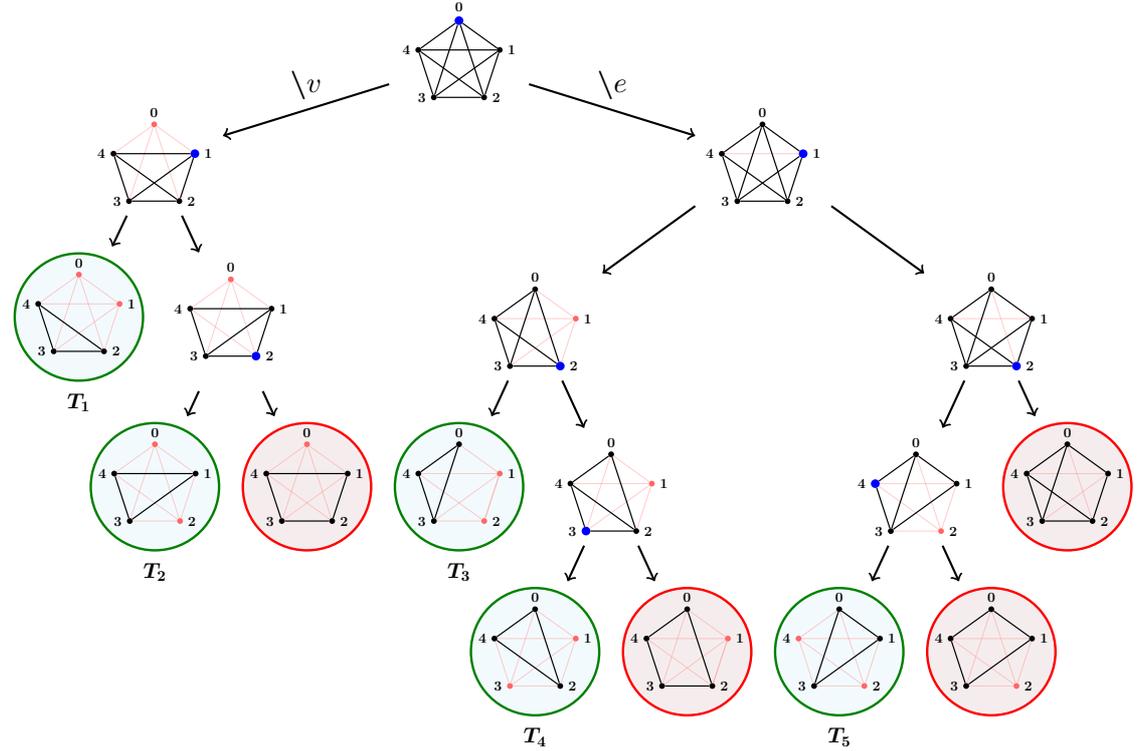
\begin{figure}[h]
\centering
\begin{tikzpicture}
 [scale=1.25,every node/.style={rectangle, draw=none, fill=none, scale=0.5625}]
\node (k5) at (0.2,0) {\begin{tikzpicture}
 [every node/.style={circle, draw=none, fill=black, scale=0.375}]
 \draw [draw=none, fill opacity=0.4] (0,0) circle (1.5cm);
  \node[label=above:{\Huge \bf 0}, fill=blue, scale=1.5] (n1) at (90:1) {};
   \node [label=right:{\Huge \bf 1}] (n2) at  (18:1) {};
  \node [label=right:{\Huge \bf 2}] (n3) at (-54:1) {};
   \node [label=left:{\Huge \bf 3}] (n4) at (-126:1) {};
   \node [label=left:{\Huge \bf 4}] (n5) at (162:1) {};
    \draw [thick] (n1)--(n2)--(n3)--(n4)--(n5)--(n1);
     \foreach \x/\y in {n1/n4,n4/n2,n2/n5,n5/n3,n1/n3}
           \draw [thick] (\x)--(\y);   
\end{tikzpicture}};
\node (k54) at (-3,-1) {\begin{tikzpicture}
 [every node/.style={circle, draw=none, fill=black, scale=0.375}]
   \node[fill=red!60, label=above:{\Huge \bf 0}] (n1) at (90:1) {};
   \node [label=right:{\Huge \bf 1}, fill=blue, scale=1.5] (n2) at  (18:1) {};
  \node [label=right:{\Huge \bf 2}] (n3) at (-54:1) {};
   \node [label=left:{\Huge \bf 3}] (n4) at (-126:1) {};
   \node [label=left:{\Huge \bf 4}] (n5) at (162:1) {};
     \foreach \x/\y in {n1/n4,n1/n3, n1/n2,n1/n5} 
         \draw [red!40, opacity=0.5] (\x)--(\y);      
    \draw [thick] (n2)--(n3)--(n4)--(n5);
     \foreach \x/\y in {n4/n2,n2/n5,n5/n3}
           \draw [thick] (\x)--(\y);   
    
\end{tikzpicture}};
\node (k551) at (3.4,-1) {\begin{tikzpicture}
 [every node/.style={circle, draw=none, fill=black, scale=0.375}]
   \node[thick, label=above:{\Huge \bf 0}] (n1) at (90:1) {};
   \node [label=right:{\Huge \bf 1}, fill=blue, scale=1.5] (n2) at  (18:1) {};
  \node [label=right:{\Huge \bf 2}] (n3) at (-54:1) {};
   \node [label=left:{\Huge \bf 3}] (n4) at (-126:1) {};
   \node [label=left:{\Huge \bf 4}] (n5) at (162:1) {};
     \foreach \x/\y in {n2/n5} 
         \draw [red!40, opacity=0.5] (\x)--(\y);      
    \draw [thick] (n2)--(n3)--(n4)--(n5);
     \foreach \x/\y in {n4/n2,n5/n3,n1/n4,n1/n3, n1/n2,n1/n5}
           \draw [thick] (\x)--(\y);   
    
\end{tikzpicture}};
\node [label=below:{\Large $\bm{T_1}$}] (k543) at (-3.8,-2.7) {\begin{tikzpicture}
 [every node/.style={circle, draw=none, fill=black, scale=0.375}]
 \draw [draw=keyg, ultra thick, fill=cyan!10, fill opacity=0.4] (0,0) circle (1.5cm);
   \node[fill=red!60, label=above:{\Huge \bf 0}] (n1) at (90:1) {};
   \node [fill=red!60, label=right:{\Huge \bf 1}] (n2) at  (18:1) {};
  \node [label=right:{\Huge \bf 2}] (n3) at (-54:1) {};
   \node [label=left:{\Huge \bf 3}] (n4) at (-126:1) {};
   \node [label=left:{\Huge \bf 4}] (n5) at (162:1) {};
     \foreach \x/\y in {n1/n4,n1/n3, n1/n2,n1/n5,n2/n3, n2/n5,n2/n4} 
         \draw [red!40, opacity=0.5] (\x)--(\y);      
    \draw [thick] (n3)--(n4)--(n5);
     \foreach \x/\y in {n3/n5}
           \draw [thick] (\x)--(\y);   
    
\end{tikzpicture}};
\node (k5441) at (-2.2,-2.75) {\begin{tikzpicture}
 [every node/.style={circle, draw=none, fill=black, scale=0.375}]
 \draw [draw=none,fill opacity=0.4] (0,0) circle (1.5cm);
   \node[fill=red!60, label=above:{\Huge \bf 0}] (n1) at (90:1) {};
   \node [label=right:{\Huge \bf 1}] (n2) at  (18:1) {};
  \node [label=right:{\Huge \bf 2}, fill=blue, scale=1.5] (n3) at (-54:1) {};
   \node [label=left:{\Huge \bf 3}] (n4) at (-126:1) {};
   \node [label=left:{\Huge \bf 4}] (n5) at (162:1) {};
     \foreach \x/\y in {n1/n4,n1/n3, n5/n3, n1/n2,n1/n5} 
         \draw [red!40, opacity=0.5] (\x)--(\y);      
    \draw [thick] (n2)--(n3)--(n4)--(n5);
     \foreach \x/\y in {n2/n5,n4/n2}
           \draw [thick] (\x)--(\y);   
    
\end{tikzpicture}};
\node [label=below:{\Large $\bm{T_2}$}] (k54413) at (-3,-4.5) {\begin{tikzpicture}
 [every node/.style={circle, draw=none, fill=black, scale=0.375}]
 \draw [draw=keyg, ultra thick, fill=cyan!10, fill opacity=0.4] (0,0) circle (1.5cm);
   \node[fill=red!60, label=above:{\Huge \bf 0}] (n1) at (90:1) {};
   \node [thick, label=right:{\Huge \bf 1}] (n2) at  (18:1) {};
  \node [fill=red!60, label=right:{\Huge \bf 2}] (n3) at (-54:1) {};
   \node [label=left:{\Huge \bf 3}] (n4) at (-126:1) {};
   \node [label=left:{\Huge \bf 4}] (n5) at (162:1) {};
     \foreach \x/\y in {n1/n4,n1/n3, n5/n3, n1/n2,n1/n5, n2/n3, n3/n4} 
         \draw [red!40, opacity=0.5] (\x)--(\y);      
    \draw [thick] (n4)--(n5);
     \foreach \x/\y in {n2/n5,n4/n2}
           \draw [thick] (\x)--(\y);   
    
\end{tikzpicture}};
\node (k544142) at (-1.4,-4.5) {\begin{tikzpicture}
 [every node/.style={circle, draw=none, fill=black, scale=0.375}]
 \draw [draw=red, ultra thick, fill=key!20, fill opacity=0.4] (0,0) circle (1.5cm);
   \node[fill=red!60, label=above:{\Huge \bf 0}] (n1) at (90:1) {};
   \node [label=right:{\Huge \bf 1}] (n2) at  (18:1) {};
  \node [label=right:{\Huge \bf 2}] (n3) at (-54:1) {};
   \node [label=left:{\Huge \bf 3}] (n4) at (-126:1) {};
   \node [label=left:{\Huge \bf 4}] (n5) at (162:1) {};
     \foreach \x/\y in {n1/n4,n1/n3, n5/n3, n1/n2,n1/n5,n4/n2} 
         \draw [red!40, opacity=0.5] (\x)--(\y);      
    \draw [thick] (n2)--(n3)--(n4)--(n5);
     \foreach \x/\y in {n2/n5}
           \draw [thick] (\x)--(\y);   
    
\end{tikzpicture}};
\node (k5514) at (1,-2.75) {\begin{tikzpicture}
 [every node/.style={circle, draw=none, fill=black, scale=0.375}]
   \node[label=above:{\Huge \bf 0}] (n2) at (90:1) {};
   \node [fill=red!60, label=right:{\Huge \bf 1}] (n1) at  (18:1) {};
  \node [label=right:{\Huge \bf 2}, fill=blue, scale=1.5] (n3) at (-54:1) {};
   \node [label=left:{\Huge \bf 3}] (n4) at (-126:1) {};
   \node [label=left:{\Huge \bf 4}] (n5) at (162:1) {};
     \foreach \x/\y in {n1/n4,n1/n3, n1/n2,n1/n5} 
         \draw [red!40, opacity=0.5] (\x)--(\y);      
    \draw [thick] (n2)--(n3)--(n4)--(n5);
     \foreach \x/\y in {n4/n2,n2/n5,n5/n3}
           \draw [thick] (\x)--(\y);   
    
\end{tikzpicture}};
\node [label=below:{\Large $\bm{T_3}$}] (k55143) at (0.2,-4.5) {\begin{tikzpicture}
 [every node/.style={circle, draw=none, fill=black, scale=0.375}]
 \draw [draw=keyg, ultra thick, fill=cyan!10, fill opacity=0.4] (0,0) circle (1.5cm);
   \node[label=above:{\Huge \bf 0}] (n2) at (90:1) {};
   \node [fill=red!60, label=right:{\Huge \bf 1}] (n1) at  (18:1) {};
  \node [fill=red!60, label=right:{\Huge \bf 2}] (n3) at (-54:1) {};
   \node [label=left:{\Huge \bf 3}] (n4) at (-126:1) {};
   \node [label=left:{\Huge \bf 4}] (n5) at (162:1) {};
     \foreach \x/\y in {n1/n4,n1/n3, n1/n3, n1/n2,n1/n5,n5/n3,n3/n4} 
         \draw [red!40, opacity=0.5] (\x)--(\y);      
    \draw [thick] (n4)--(n5);
     \foreach \x/\y in {n4/n2,n2/n5}
           \draw [thick] (\x)--(\y);   
    
\end{tikzpicture}};
\node (k551441) at (1.8,-4.5) {\begin{tikzpicture}
 [every node/.style={circle, draw=none, fill=black, scale=0.375}]
   \node[label=above:{\Huge \bf 0}] (n2) at (90:1) {};
   \node [fill=red!60, label=right:{\Huge \bf 1}] (n1) at  (18:1) {};
  \node [label=right:{\Huge \bf 2}] (n3) at (-54:1) {};
   \node [label=left:{\Huge \bf 3}, fill=blue, scale=1.5] (n4) at (-126:1) {};
   \node [label=left:{\Huge \bf 4}] (n5) at (162:1) {};
     \foreach \x/\y in {n1/n4,n1/n3,n4/n2, n1/n2,n1/n5} 
         \draw [red!40, opacity=0.5] (\x)--(\y);      
    \draw [thick] (n2)--(n3)--(n4)--(n5);
     \foreach \x/\y in {n2/n5,n5/n3}
           \draw [thick] (\x)--(\y);   
    
\end{tikzpicture}};
\node [label=below:{\Large $\bm{T_4}$}] (k5514413) at (1,-6.25) {\begin{tikzpicture}
 [every node/.style={circle, draw=none, fill=black, scale=0.375}]
 \draw [draw=keyg, ultra thick, fill=cyan!10, fill opacity=0.4] (0,0) circle (1.5cm);
   \node[label=above:{\Huge \bf 0}] (n2) at (90:1) {};
   \node [fill=red!60, label=right:{\Huge \bf 1}] (n1) at  (18:1) {};
  \node [label=right:{\Huge \bf 2}] (n3) at (-54:1) {};
   \node [fill=red!60, label=left:{\Huge \bf 3}] (n4) at (-126:1) {};
   \node [label=left:{\Huge \bf 4}] (n5) at (162:1) {};
     \foreach \x/\y in {n1/n4,n1/n3,n4/n2, n1/n2,n1/n5,n4/n3,n4/n5} 
         \draw [red!40, opacity=0.5] (\x)--(\y);      
    \draw [thick] (n2)--(n3);
     \foreach \x/\y in {n2/n5,n5/n3}
           \draw [thick] (\x)--(\y);   
    
\end{tikzpicture}};
\node (k55144142) at (2.6,-6.25) {\begin{tikzpicture}
 [every node/.style={circle, draw=none, fill=black, scale=0.375}]
 \draw [draw=red, ultra thick, fill=key!20, fill opacity=0.4] (0,0) circle (1.5cm);
   \node[label=above:{\Huge \bf 0}] (n2) at (90:1) {};
   \node [fill=red!60, label=right:{\Huge \bf 1}] (n1) at  (18:1) {};
  \node [label=right:{\Huge \bf 2}] (n3) at (-54:1) {};
   \node [label=left:{\Huge \bf 3}] (n4) at (-126:1) {};
   \node [label=left:{\Huge \bf 4}] (n5) at (162:1) {};
     \foreach \x/\y in {n1/n4,n1/n3,n4/n2, n1/n2,n1/n5,n5/n3,n3/n1,n3/n4} 
         \draw [red!40, opacity=0.5] (\x)--(\y);      
    \draw [thick] (n2)--(n3)--(n4)--(n5);
     \foreach \x/\y in {n2/n5}
           \draw [thick] (\x)--(\y);   
    
\end{tikzpicture}};
\node (k55152) at (5.8,-2.75) {\begin{tikzpicture}
 [every node/.style={circle, draw=none, fill=black, scale=0.375}]
   \node[label=above:{\Huge \bf 0}] (n1) at (90:1) {};
   \node [label=right:{\Huge \bf 1}] (n2) at  (18:1) {};
  \node [label=right:{\Huge \bf 2}, fill=blue, scale=1.5] (n3) at (-54:1) {};
   \node [label=left:{\Huge \bf 3}] (n4) at (-126:1) {};
   \node [label=left:{\Huge \bf 4}] (n5) at (162:1) {};
     \foreach \x/\y in {n2/n5,n1/n3} 
         \draw [red!40, opacity=0.5] (\x)--(\y);      
    \draw [thick] (n2)--(n3)--(n4)--(n5);
     \foreach \x/\y in {n4/n2,n5/n3,n1/n4, n1/n2,n1/n5}
           \draw [thick] (\x)--(\y);   
    
\end{tikzpicture}};
\node (k5515253) at (6.6,-4.5) {\begin{tikzpicture}
 [every node/.style={circle, draw=none, fill=black, scale=0.375}]
 \draw [draw=red, ultra thick, fill=key!20, fill opacity=0.4] (0,0) circle (1.5cm);
   \node[label=above:{\Huge \bf 0}] (n1) at (90:1) {};
   \node [label=right:{\Huge \bf 1}] (n2) at  (18:1) {};
  \node [label=right:{\Huge \bf 2}] (n3) at (-54:1) {};
   \node [label=left:{\Huge \bf 3}] (n4) at (-126:1) {};
   \node [label=left:{\Huge \bf 4}] (n5) at (162:1) {};
     \foreach \x/\y in {n4/n2,n2/n5,n1/n3} 
         \draw [red!40, opacity=0.5] (\x)--(\y);      
    \draw [thick] (n2)--(n3)--(n4)--(n5);
     \foreach \x/\y in {n5/n3,n1/n4, n1/n2,n1/n5}
           \draw [thick] (\x)--(\y);   
    
\end{tikzpicture}};
\node (k5515241) at (5,-4.5) {\begin{tikzpicture}
 [every node/.style={circle, draw=none, fill=black, scale=0.375}]
   \node[thick, label=above:{\Huge \bf 0}] (n1) at (90:1) {};
   \node [thick,label=right:{\Huge \bf 1}] (n2) at  (18:1) {};
  \node [fill=red!60, label=right:{\Huge \bf 2}] (n3) at (-54:1) {};
   \node [label=left:{\Huge \bf 3}] (n4) at (-126:1) {};
   \node [label=left:{\Huge \bf 4}, fill=blue, scale=1.5] (n5) at (162:1) {};
     \foreach \x/\y in {n2/n5,n1/n3,n5/n3,n3/n4,n2/n3} 
         \draw [red!40, opacity=0.5] (\x)--(\y);      
    \draw [thick] (n4)--(n5);
     \foreach \x/\y in {n4/n2,n1/n4, n1/n2,n1/n5}
           \draw [thick] (\x)--(\y);   
    
\end{tikzpicture}};
\node [label=below:{\Large $\bm{T_5}$}] (k55152413) at (4.2,-6.25) {\begin{tikzpicture}
 [every node/.style={circle, draw=none, fill=black, scale=0.375}]
 \draw [draw=keyg, ultra thick, fill=cyan!10, fill opacity=0.4] (0,0) circle (1.5cm);
   \node[label=above:{\Huge \bf 0}] (n1) at (90:1) {};
   \node [label=right:{\Huge \bf 1}] (n2) at  (18:1) {};
  \node [fill=red!60, label=right:{\Huge \bf 2}] (n3) at (-54:1) {};
   \node [label=left:{\Huge \bf 3}] (n4) at (-126:1) {};
   \node [fill=red!60, label=left:{\Huge \bf 4}] (n5) at (162:1) {};
     \foreach \x/\y in {n2/n5,n1/n3,n5/n3,n1/n5,n4/n5,n2/n3,n3/n4} 
         \draw [red!40, opacity=0.5] (\x)--(\y);      
     \foreach \x/\y in {n4/n2,n1/n4, n1/n2}
           \draw [thick] (\x)--(\y);   
    
\end{tikzpicture}};
\node (k551524142) at (5.8,-6.25) {\begin{tikzpicture}
 [every node/.style={circle, draw=none, fill=black, scale=0.375}]
 \draw [draw=red, ultra thick, fill=key!20, fill opacity=0.4] (0,0) circle (1.5cm);
   \node[label=above:{\Huge \bf 0}] (n1) at (90:1) {};
   \node [label=right:{\Huge \bf 1}] (n2) at  (18:1) {};
  \node [fill=red!60, label=right:{\Huge \bf 2}] (n3) at (-54:1) {};
   \node [label=left:{\Huge \bf 3}] (n4) at (-126:1) {};
   \node [label=left:{\Huge \bf 4}] (n5) at (162:1) {};
     \foreach \x/\y in {n2/n5,n1/n3,n5/n3,n1/n4,n2/n3,n3/n4} 
         \draw [red!40, opacity=0.5] (\x)--(\y);      
    \draw [thick] (n4)--(n5);
     \foreach \x/\y in {n4/n2,n1/n2,n1/n5}
           \draw [thick] (\x)--(\y);   
    
\end{tikzpicture}};

\foreach \x/\y in {k54/k543, k54/k5441, k5441/k54413, k5441/k544142,k5514/k55143, k5514/k551441, k551441/k5514413, k551441/k55144142, k55152/k5515241,k55152/k5515253, k5515241/k55152413, k5515241/k551524142}
           \draw [thick, ->] (\x)--(\y);     
   \foreach \x/\y in {k551/k5514, k551/k55152}
                \draw [thick, ->] (\x)--(\y);    
                
\draw [thick, ->] (k5)--(k54) node [pos=0.5, above, scale=1.75] {$\bm \setminus v$};
\draw [thick, ->] (k5)--(k551) node [pos=0.5, above, scale=1.75] {$\bm \setminus e$};
%
\end{tikzpicture}
\caption{Full reduction tree for the convex pentagon. Blue color indicates the working on node. Relabelings are not marked. Each of its 5 triangulations is encoded in a different green leaf.}
\label{fig:k5fullred}
\end{figure}

Figure \ref{fig:k5fullred} shows the full reduction tree for the convex $K_5$. To give a for instance, consider the green leaf encoding triangulation $T_3$. $T_3$ includes the triangle $v_0v_3v_4$ (end of reduction); as a left child of its parent node, which branched w.r.t.\ $v_2$/$v_0v_3$, it includes $v_0v_3$ as an span-2 edge of the parent $K_4$ subgraph, so triangle $v_0v_2v_3$ is in $T_4$. The $K_4$ parent node is a left child of the $K_{5,-1}$ node which branched w.r.t.\ $v_1$, therefore the triangulation does also include triangle $v_0v_1v_2$; finally (though not necessary), the $K_{5,-1}$ node being a right child of the $K_5$ which branched w.r.t.\ $v_0$ suggests that $v_4v_1$ is missing from the $T_4$. 

Note that this is a small example, and some leaves ($T_2$, $T_4$, $T_5$) already indicate unambiguously the encoded triangulation. Observe also that even if it is not directly visible that $T_1$ differs from $T_4$, our simple branching rule guarantees that $T_1$ includes $v_4v_1$ as an edge, while $T_4$ does not, as they belong to a different subtree of the root node which branched w.r.t.\ $v_0$/$v_4v_1$.

\end{document}